\documentclass[journal,10pt,draftclsnofoot,onecolumn]{IEEEtran}
\usepackage{subfigure}
\usepackage{setspace}
\usepackage{multirow}
\usepackage{amsmath}
\usepackage{amssymb}
\usepackage{amsfonts}
\usepackage{amscd}
\usepackage{mathrsfs}
\usepackage[final]{graphicx}
\usepackage{graphicx}
\usepackage{psfrag}
\usepackage{color}
\usepackage{url}

\usepackage{cite}
\usepackage{amssymb}

\newtheorem{theorem}{Theorem}
\newtheorem{lemma}{Lemma}
\newtheorem{definition}{Definition}

\newtheorem{proposition}{Proposition}
\newtheorem{corollary}{Corollary}
\newtheorem{example}{Example}
\newtheorem{remark}{Remark}

\begin{document}

\title{A Singleton Bound for Lattice Schemes}
\author{Srikanth~B.~Pai,
        B.~Sundar~Rajan,~\IEEEmembership{Fellow,~IEEE}
\thanks{The authors are with the Department
of Electrical and Communication Engineering, Indian Institute of Science, Bangalore-560012, India. E-mail: srikanthbpai@gmail.com, bsrajan@ece.iisc.ernet.in.}}

\maketitle
\IEEEpeerreviewmaketitle

\begin{abstract}
In this paper, we derive a Singleton bound for lattice schemes and obtain Singleton bounds known for binary codes and subspace codes as special cases. It is shown that the modular structure affects the strength of the Singleton bound. We also obtain a {\it new} upper bound on the code size for non-constant dimension codes. The plots of this bound along with plots of the code sizes of known non-constant dimension codes in the literature reveal that our bound is tight for certain parameters of the code.
\let\thefootnote\relax\footnotetext{\emph{Key words and phrases}: Binary codes, network codes, posets, modular lattices, metric spaces, Singleton bound.}
\let\thefootnote\relax\footnotetext{Parts of the content in this paper has been presented in IEEE International Symposium on Information Theory 2013, Istanbul.}
\end{abstract}

\section{Introduction}

In \cite{KoeKschi}, a model for error correction for random network coding is proposed. In the proposed model, subspaces of a vector space are transmitted and the network acts like a channel which corrupts the transmitted subspace and a different subspace can be possibly received. The codes constructed in this framework are called \emph{subspace codes}. In \cite{KoeKschi}, only codes where all subspaces have the same dimension are considered. Further, a Singleton bound, and sphere packing bounds are derived and a code construction is presented. In this paper, we generalize the results presented in \cite{KoeKschi}.

Let $q$ be a prime power, $n$ a natural number and $W$ a $n$-dimensional vector space over $\mathbb{F}_q$. The class of all subspaces of $W$, denoted by ${\cal P}(W)$ and called \emph{a projective space}, is the input alphabet and the output alphabet for the operator channel \cite{KoeKschi}. This channel is used to model errors and erasures arising in network communication while employing random network coding (RNC). A channel encoder maps incoming messages to subspaces. Subspaces are transmitted through the operator channel and received by the sink. The operator channel model for RNC is shown in Fig. \ref{fig_RNCop} where $V$ is a transmitted subspace and $U$ is a received subspace, $k=\dim(U \cap V)$ and $E$ is called the error subspace. We say that the operator channel introduces $\rho = \dim(V) - k$ erasures and $t = \dim(E)$ errors. A \emph{subspace code} is a subset of ${\cal P}(W)$. A metric, called the \emph{subspace distance} $d_S$, is defined on a projective space in \cite{KoeKschi}. Given two elements $A, B \in {\cal P}(W)$, \emph{subspace distance} is defined as 
\begin{equation} d_S(A,B):=\dim(A+B) - \dim(A \cap B). \end{equation}

\begin{figure}
 \centering
 \includegraphics[scale=0.5, trim=50 225 40 80, clip=true]{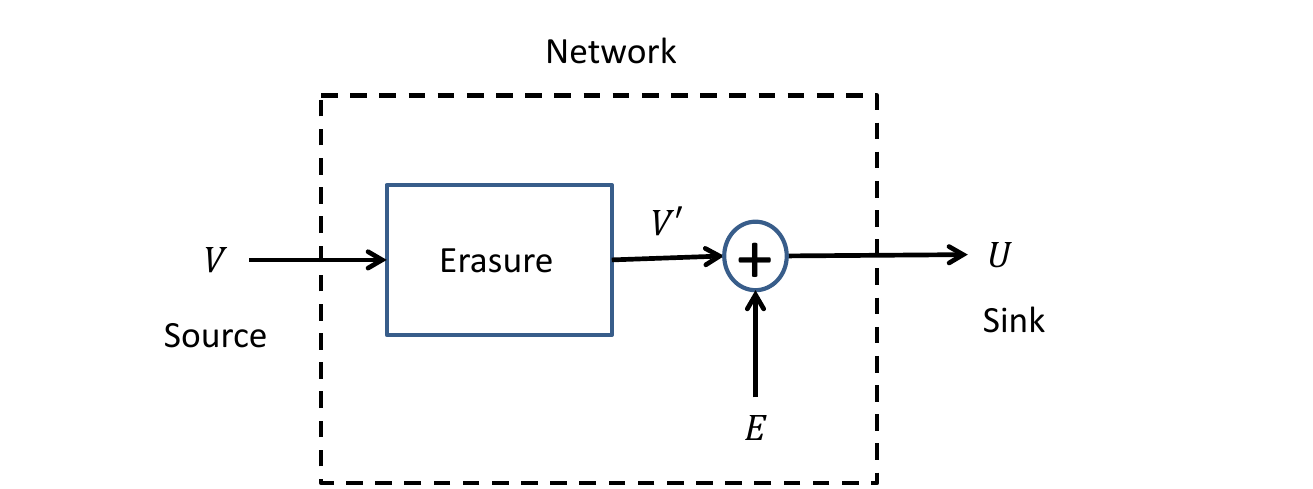}
 \caption{The operator channel for random network coding}
 \label{fig_RNCop}
\end{figure}

On the other hand, in binary coding, binary vectors are transmitted and Hamming distance serves as a metric on a binary vector space. It was noted in \cite{Braun} that the model for subspace codes is similar to the model of binary codes. The similarity is captured using lattices. The author of \cite{Braun} points out that the lattice of subsets can be used to represent binary codes and the lattice of subspaces can be used for subspace codes. Motivated by this observation, we define the concept of `lattice schemes' in this paper, which serves as an analogue for both binary codes and subspace codes. A Singleton bound for constant height codes in a modular lattice was derived in \cite{KendSch}. The technique of `puncturing a code' used to prove Singleton bound in that paper is very similar to the one used in \cite{KoeKschi}. It is shown that every time one `punctures' a code, the minimum distance of the code can drop by at most two units. In this paper, we generalize the Singleton bound to non-constant height lattice schemes in a geometric modular lattice. It turns out that the Singleton bound derived for constant dimension codes in \cite{KoeKschi} is not tight \cite{EtzVar}. On the other hand, the classical Singleton bound for binary codes is tight at least in certain cases. Our main motivation of this paper was to investigate this difference in behavior. The binary codes are schemes in a distributive lattice and the subspace codes are schemes in a modular lattice. We show that the lack of distributivity in subspace codes causes the Singleton bound to weaken. 

Our contributions in this paper can be summarized as follows:
\begin{enumerate}
\item We modify the definition of puncturing given in \cite{KendSch}. Under the modified definition of puncturing, we show that the drop in minimum distance after puncturing a scheme in a geometric modular lattice can be at most two units. However, we show that the drop in minimum distance after puncturing a scheme in a geometric distributive lattice is at most one unit.
\item According to a theorem in lattice theory, any distributive lattice cannot have the lattice $M_3$, shown in Fig. \ref{fig_Mthree}, as its sub-lattice \cite[Chapter 2, Theorem 13]{Birkhoff}. We show that the minimum distance of a lattice scheme drops by two units after puncturing any non-distributive modular geometric lattice and this is due to the $M_3$ structure. 
\item We derive a Singleton bound for lattice schemes in a geometric modular lattice (called the `Lattice Singleton Bound'). Lattice schemes is not assumed to be of a constant height as in \cite{KendSch}. The classical Singleton bound is derived as a corollary.
\item Using the Lattice Singleton Bound, we derive a bound on the code size for non-constant dimension codes in a projective space. To the best of our knowledge, this upper bound on code size is the first upper bound for non-constant dimension codes in the literature.
\end{enumerate}

The paper is organized as follows: In Section \ref{sec_lat_prelims} we give an introduction to lattices. Section \ref{sec_lat_sch} defines a lattice scheme and shows that binary codes and subspace codes are equivalent to lattice schemes for some lattices. We modify the definition of puncturing in \cite{KendSch} in Section \ref{sec_main_res}. A Singleton bound for geometric modular lattice is then derived and its applications are presented in the same section.

\begin{figure}
 \centering
  \includegraphics[scale=0.5, trim=90 150 75 90, clip=true]{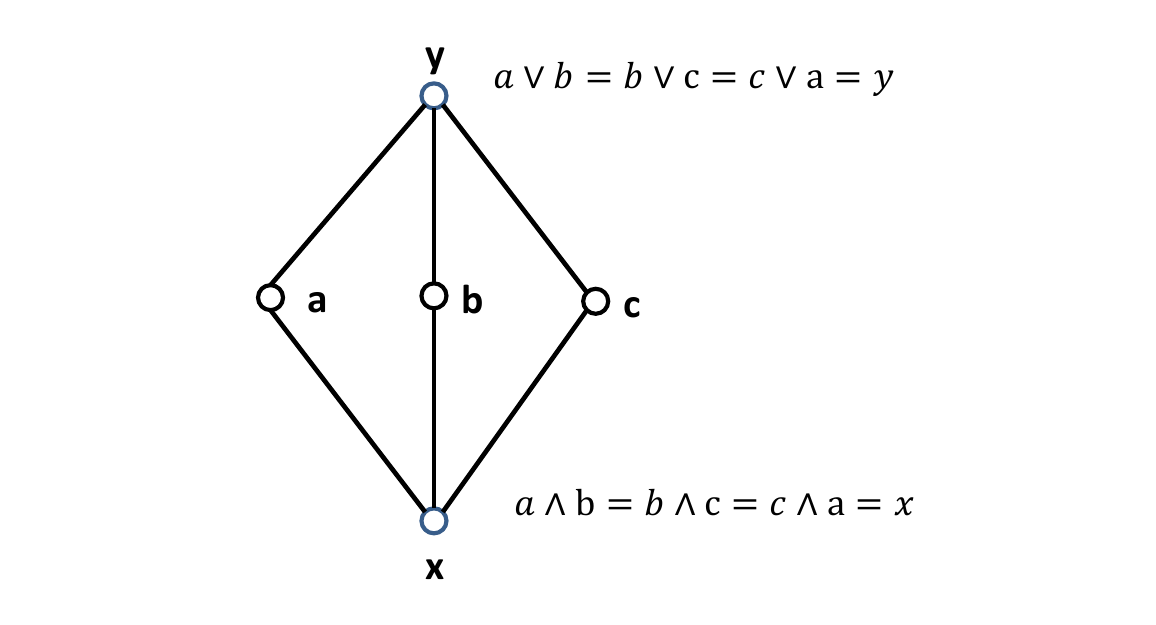}
  \caption{The $M_3$ lattice}
 \label{fig_Mthree}
\end{figure}

{\it Notations:} A set is denoted by a capital letter and its elements will be denoted by small letters (For example, $ x \in X $). All the sets considered in this paper will be finite. Given a set $X$, $|X|$ denotes the number of elements in the set and for a subset $A$ of $X$, $A^c$ denotes the complement of the set $A$ in $X$. For two sets $A$ and $B$, $A \times B$ denotes the Cartesian product of the two sets, i.e. $A \times B := \{(a,b) | a \in A, b \in B\}$. $A \triangle B$ represents the symmetric difference of sets, i.e. $A \triangle B := (A^c \cap B) \cup (A \cap B^c)$. A lattice will be denoted by $(L,\vee,\wedge)$ and sometimes we will drop the join and the meet notation, simply calling it $L$. $\mathbb{F}_q$ denotes the finite field with $q$ elements where $q$ is a power of a prime number. $\mathbb{F}_q^*$ denotes all the non zero elements of $\mathbb{F}_q$. The symbol $V$ denotes a vector space (generally over $\mathbb{F}_q$). For a subset $S$ of $V$, $\langle S \rangle$ denotes the linear span of all the elements in $S$. Given two subspaces $A$ and $B$, $A+B$ denotes the smallest subspace containing both $A$ and $B$. Let $V$ be a $n$ dimensional space over $\mathbb{F}_q$. Then the symbol ${\cal G}(n,l)$ denotes the Grassmanian, i.e. the set of all $l$ dimensional subspaces of $V$. The number of elements in ${\cal G}(n,l)$ is denoted as ${n \brack l}_q$. $\mathbb{F}_q^n$ denotes the $n$ dimensional vector space of $n$-tuples over $\mathbb{F}_q$. Given a vector $x \in \mathbb{F}_q^n$, $x_i$ denotes the $i$-th co-ordinate of $x$. The support of $x \in \mathbb{F}_q^n$ (denoted by $\text{support}\{x\}$) is defined as the set of indices where the vector is non-zero. $(\mathbb{F}_q^n,d_H)$ denotes the vector space $\mathbb{F}_q^n$ with the Hamming metric, i.e. $d_H(a,b) = |\text{support}\{a-b\}|$.

\section{Introduction to Lattices}
\label{sec_lat_prelims}

This section serves as a quick introduction to lattice theory. All
the lattice theory definitions and theorems required for the rest
of the paper are given in this section. We follow notations and definitions from \cite{Birkhoff}.

\begin{definition} \label{def_Poset_order} A \textit{poset} is a
pair $(P,\leq)$, where $P$ is a set and $\leq$ is a binary relation (called the \textit{order relation}) on the set $P$ satisfying: 
\begin{enumerate}
\item {\small (Reflexivity)} For all ${\displaystyle x,x\leq x}.$
\item {\small (Antisymmetry)} If ${\displaystyle x\leq y,y\leq x}$, then
$x=y.$
\item {\small (Transitivity)} If ${\displaystyle x\leq y,y\leq z}$, then
$x\leq z.$ 
\end{enumerate}
\end{definition}

For the remainder of the paper, $P$ denotes a poset with $\leq$
as the order relation. If $x\leq y$ and $x\neq y$, then we use the
shorthand $x<y$. $x\leq y$ is read as {}``x is less than y'' or
{}``x is contained in y''. 

\begin{definition} An \textit{upper bound} (\textit{lower bound})
of a subset $X$ of $P$ is an element $a\in P$ containing (contained
in) every $x\in X$. The \textit{least upper bound} (\textit{greatest
lower bound}) of $X$ is the element of $P$ contained in (containing)
every upper bound (lower bound) of $X$. \end{definition}

If a least upper bound, or a greatest lower bound of a set exists,
it is unique due to the antisymmetry property of the order relation (Definition \ref{def_Poset_order}). The least upper bound of a set
$X$ is denoted by sup $X$ and the greatest lower bound is denoted
by inf $X$.

\begin{definition} A \textit{lattice} $L$ is a poset which has the
property that ${\displaystyle \forall a,b\in L}$, the sup$\{a,b\}$
exists and inf$\{a,b\}$ exists. The sup$\{a,b\}$ is denoted by $a\vee b$ (read
as {}``a \textit{join} b'') and the inf$\{a,b\}$ is denoted by
$a\wedge b$ (read as {}``a \textit{meet} b''). The lattice itself
is denoted by $(L,\vee,\wedge)$. \end{definition}

We will assume, for the purposes of the paper, that all the lattices are finite and
have a unique greatest element denoted by $I$, and a unique least
element denoted by $O$.

\begin{definition} A \textit{sub-lattice} of a lattice $L$ is a subset
$K$ of $L$ that satisfies the following condition:
\[
a,b\in K\implies a\vee b\in K, a\wedge b\in K
\]
\end{definition}

\begin{definition} A map $\psi$ from $L$ to $K$ is said to be
a \textit{lattice homomorphism} if it satisfies the following conditions: 
\begin{enumerate}
\item $\psi(a\vee b)=\psi(a)\vee\psi(b)$ 
\item $\psi(a\wedge b)=\psi(a)\wedge\psi(b)$ 
\end{enumerate}
Further, if $\psi$ is bijective, then we say that $L$ and $K$ are
\textit{isomorphic}. \end{definition}

\begin{example} Let $X=\{1,2,3\}$ and ${\cal P}(X)$ denote the
power set of $X$. We can view $({\cal P}(X),\subseteq)$ as a poset
where set inclusion is the order relation. The power set
under this order is a lattice. For any subsets $A$ and $B$,
$A\vee B=A\cup B$ and $A\wedge B=A\cap B$. Further, the lattice generated
by subsets of $\{1,2\}$ is a sub-lattice. Notice that $I=X$ and $O=\phi$. 
\end{example}

One can completely specify an order of a poset (finite ones) by a
\textit{Hasse diagram}, like the one shown in Fig. \ref{fig_Power_set}.
If $a\leq b$ and there is no $t$ such that $a\leq t\leq b$, we
say that {}``$b$ \textit{covers} $a$''. In the Hasse diagram
of a lattice, $a$ and $b$ are joined iff $b$ covers $a$ or $a$
covers $b$. The diagram is drawn in such a way that if $b$ covers
$a$, then $b$ is written above $a$. Clearly, $a\leq b$ iff there
exists a path from $a$ moving up to $b$. The order relation is completely
specified by such a diagram. Naturally, $I$ will be the topmost element
and $O$ will be the lowest element.

\begin{figure}
 \centering
 \includegraphics[scale=0.4, trim=20 10 50 10, clip=true]{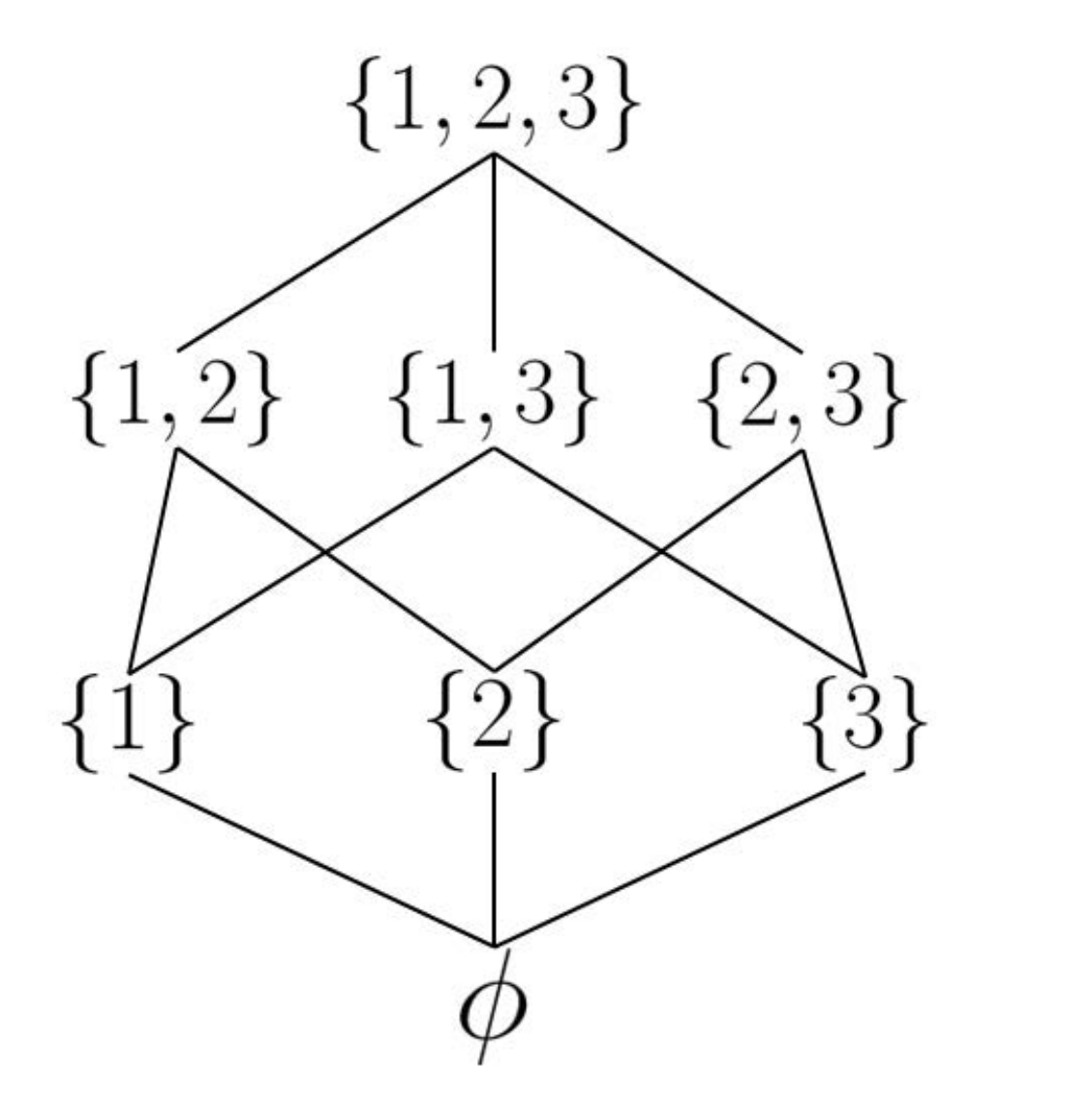}
 \caption{{\it Pow}($\{1,2,3\}$): A lattice of subsets}
 \label{fig_Power_set}
\end{figure}

For a set $X$, $($\textit{Pow($X$)}$,\cup,\cap)$ will denote lattice
of the power set of $X$ with union and intersection as join and meet
of the lattice respectively. 

\begin{definition} A lattice $(L,\vee,\wedge)$ is \textit{distributive}
if the following conditions hold: 
\begin{enumerate}
\item For all $a,b,c\in L,a\vee(b\wedge c)=(a\vee b)\wedge(a\vee c).$ 
\item For all $a,b,c\in L,a\wedge(b\vee c)=(a\wedge b)\vee(a\wedge c).$ 
\end{enumerate}
\end{definition}

It is immediate that for any set $X$, $($\textit{Pow}($X$)$,\cup,\cap)$
is a distributive lattice. However all lattices are not distributive,
as seen in the examples below.

\begin{example} Let $V$ be a vector space of dimension $n$ over
a field $\mathbb{F}_{q}$. The class of all the subspaces of $V$,
denoted by $\text{Sub}(V)$, can be ordered under inclusion in a manner
similar to a set. The join of two subspaces $A$ and $B$ will then
be the smallest subspace containing both $A$ and $B$. This means
that $A\vee B=A+B$ and the largest subspace contained in both $A$
and $B$ is $A\cap B$. So the meet of $A$ and $B$ is $A\cap B$.
Such a lattice will be denoted as $(\text{Sub}(V),+,\cap)$. Clearly
$I=V$ and $O=\{0_{V}\}$. This lattice will be called the \textit{projective lattice}.

This lattice is not distributive. To see this, we consider the vector
space $V=\mathbb{F}_{2}^{2}$ over $\mathbb{F}_{2}$.

The Hasse diagram of $(\text{Sub}(\mathbb{F}_{2}^{2}),+,\cap)$ is
shown in Fig. \ref{fig_Subspace}. Here, $A=\langle\{(0,1)\}\rangle$,
$B=\langle\{(1,0)\}\rangle$ and $C=\langle\{(1,1)\}\rangle$. Clearly, 
\[
A = A\cap(B+C)\neq A\cap B+A\cap C = \langle\{(0,0)\}\rangle
\]
 and thus $(\text{Sub}(\mathbb{F}_{2}^{2}),+,\cap)$ is not distributive.
$(\text{Sub}(\mathbb{F}_{2}^{2}),+,\cap)$ is identified by the name
\textit{M$_{3}$}.  \end{example}

\begin{figure}
 \centering
 \includegraphics[scale=0.4, trim=80 80 65 80, clip=true]{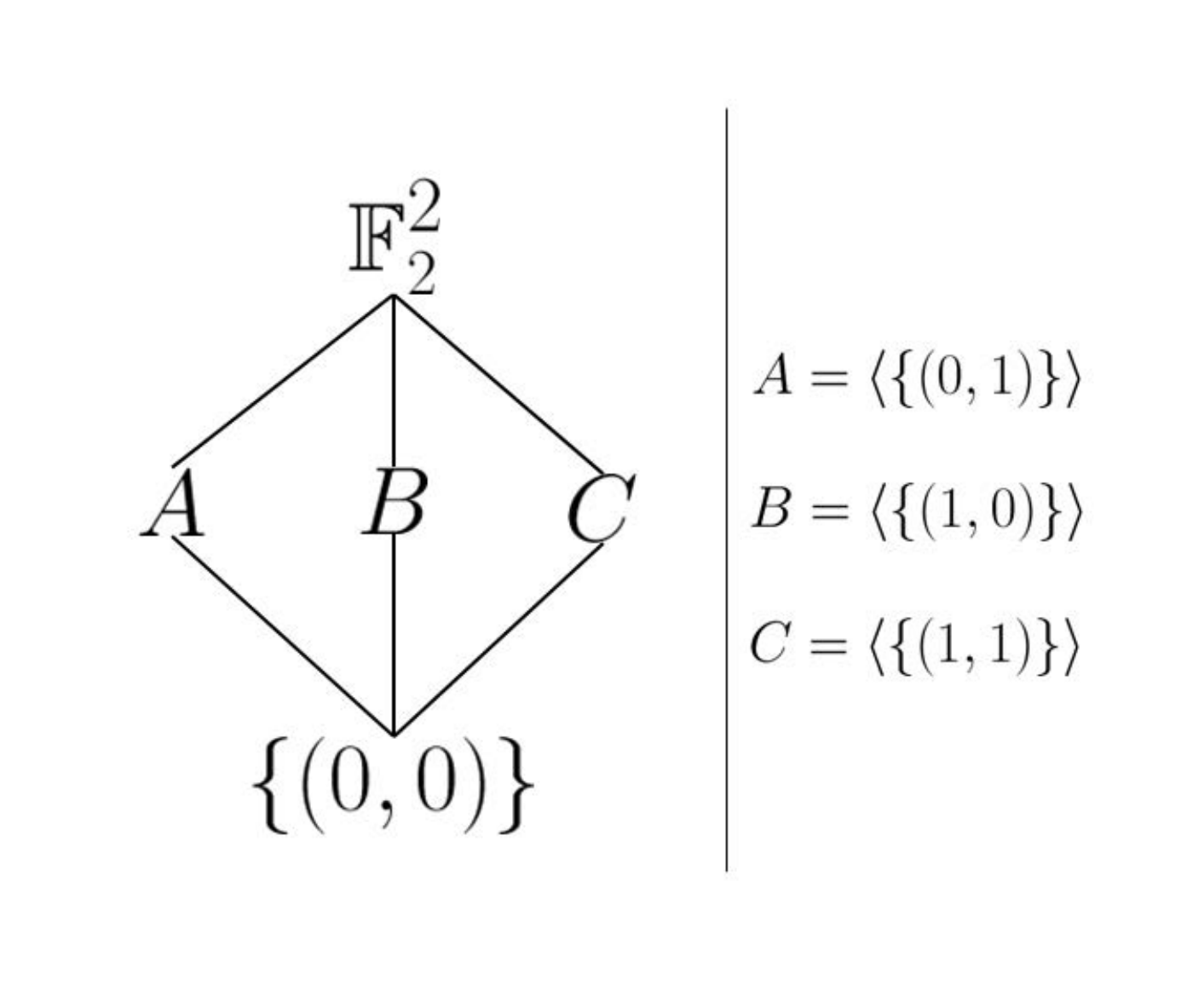}
 \caption{$\text{Sub}(\mathbb{F}_2^2)$: A lattice of subspaces, also called $M_3$}
 \label{fig_Subspace}
\end{figure}

The above lattice is not distributive, i.e. $A\cap(B+C)=A\cap B+A\cap C$
fails for three distinct subspaces $A,B,C$. If $A\subseteq C$, one
can prove that $A+(B\cap C)=(A+B)\cap C$ and thus it behaves partly
distributively. This property is captured in the following definition:

\begin{definition} A lattice $(L,\vee,\wedge)$ is \textit{modular}
if the following condition holds:\\
 (Modularity) If $a\leq c$, then $a\vee(b\wedge c)=(a\vee b)\wedge c.$
\end{definition}

If the lattice is distributive, the modularity condition holds. This
means that \textit{any distributive lattice is modular}. However $(\text{Sub}(V),+,\cap)$
is a modular lattice that is non distributive. It turns out that the $M_3$ lattice characterizes modular lattices.

\begin{theorem}\cite[pg.39, Th. 13]{Birkhoff}
\label{thm_mthree}
Any modular nondistributive lattice contains a sub-lattice isomorphic to $M_3$.
\end{theorem}

The observation that for $($\textit{Pow($X$)}$,\cup,\cap)$, any
two elements $A,B$ satisfy $|A\cup B|+|A\cap B|=|A|+|B|$, and in
an analogous fashion, for $(\text{Sub}(V),+,\cap)$, any two elements
$A,B$ satisfy $\text{dim}(A + B)+\text{dim}(A\cap B)=\text{dim}(A)+\text{dim}(B)$,
seems to suggest that modular lattices must satisfy an equation of
the form $v(a\vee b)+v(a\wedge b)=v(a)+v(b)$ for some real-valued
function $v$ on the lattice. This is indeed true and in turn, such
a function helps characterize modular lattices. We need the following
few definitions and results to make the characterization precise.

\begin{definition} An \textit{isotone valuation} on a lattice $L$
is a real valued function $v$ on $L$ that satisfies: 
\begin{enumerate}
\item {\small (Valuation)} For all $x,y\in L$, \[{\displaystyle v(x\vee y)+v(x\wedge y)=v(x)+v(y)}.\]
\item {\small (Isotone)} $x\leq y\implies v(x)\leq v(y)$ 
\end{enumerate}
\end{definition}

Additionally, the isotone valuation is called \textit{positive}, if
$x<y\implies v(x)<v(y)$. $d_{v}:=v(a\vee b)-v(a\wedge b)$ is said
to be the distance induced by $v$.

\begin{theorem}\cite[pg.230, Th.1]{Birkhoff} \label{thm_pseudo_metric} Given a lattice $L$ and
an isotone valuation $v$, the function $d_{v}(a,b):=v(a\vee b)-v(a\wedge b)$
is a metric iff $v$ is positive. In general, $d_{v}$ is a pseudo
metric. \end{theorem}

In a lattice $(L,\vee,\wedge)$, a \textit{chain} $C$ is a subset
of $L$ with the property that for all $a,b\in C$, $a\leq b$ or
$b\leq a$. We say that in a chain, any two elements are comparable.
Given two elements $a,b\in L$, a chain $\{x_{1},x_{2},\cdots,x_{l}\}$
of $L$ with the property $a=x_{0}<x_{1}<\cdots<x_{l}=b$ is called
a \textit{chain between $a$ and $b$}. The \textit{length} of the
chain is defined as $l$.

\begin{definition} The \textit{height} of an element $x$ in a lattice $L$
is the maximum length of a chain between $O$
and $x$. It is denoted by $h_{L}(x)$. \end{definition}

The number $h_{L}(I)$ is called \textit{the height of the lattice}
$L$ or the \textit{dimension of} $L$. Note that the chain from $O$ to $O$ contains only one element, and thus $h_L(O) = 0$.
We need an additional property to characterize modular lattices.

\begin{definition} A lattice is said to have the \textit{Jordan-Dedekind
property} if all maximal chains between two elements have the same
finite length. \end{definition}
All lattices need not have the Jordan-Dedekind property. For example, in the $N_5$ lattice, shown in Fig.\ref{fig_Hasse_diag_N_5}, there are two maximal chains from $d$ to $u$. The maximal chain $d \to a \to b \to u$ has three units of length and the other maximal chain $d \to c \to u$ has a length of two units. Thus $N_5$ does not satisfy the Jordan Dedekind property. It turns out that modularity
is closely related to this property.

\begin{figure}
 \centering
 \includegraphics[scale=0.4, trim=10 80 10 10, clip=true]{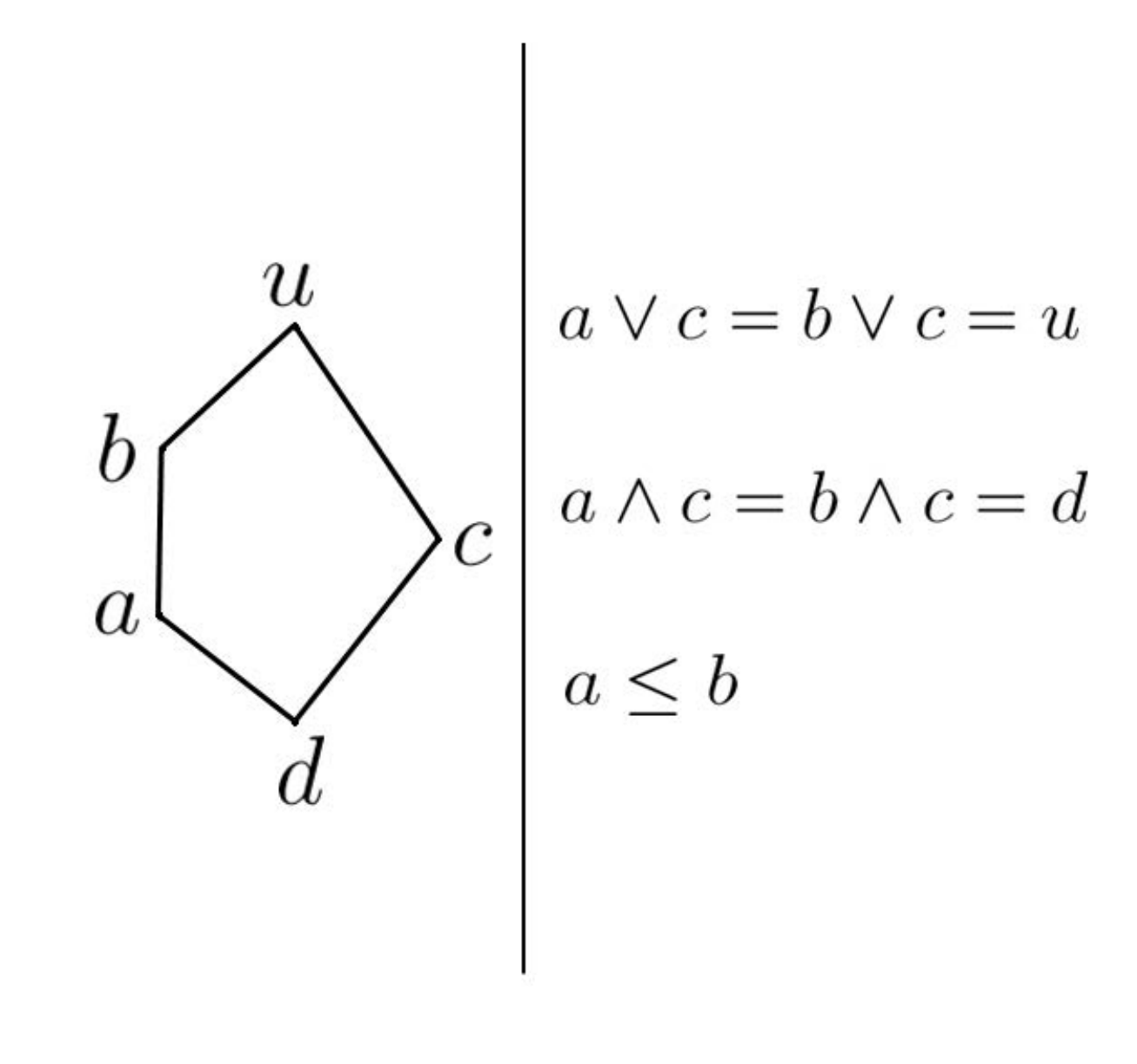}
 \caption{The $N_5$ lattice: A lattice without Jordan-Dedekind property}
 \label{fig_Hasse_diag_N_5}
\end{figure}

\begin{theorem} \label{JD} \cite[pg.41, Th.16]{Birkhoff} Let $(L,\vee,\wedge)$ be a lattice of finite length
with the height function $h$, then the following conditions are equivalent: 
\begin{enumerate}
\item $L$ is a modular lattice. 
\item $L$ has the Jordan-Dedekind property and $h$ is a valuation. 
\end{enumerate}
\end{theorem}

The $N_5$ lattice shown in Fig. \ref{fig_Hasse_diag_N_5} does not satisfy the Jordan-Dedekind property and therefore, due to Theorem \ref{JD}, $N_5$ is non-modular.

The height function $h$ is a positive isotone evaluation for modular lattices.
Therefore from Theorem \ref{thm_pseudo_metric}, the distance function
induced by $h$ is a metric on the modular lattice. All the lattices,
considered in this paper, will be modular. We will use this observation
in the next section to define a coding metric for the definition of
lattice schemes. We, therefore, have the following proposition which follows from Theorem \ref{thm_pseudo_metric}.

\begin{theorem} \label{thm_height_metric} Let $L$ be a modular lattice of finite length with
the height function $h$, then $d_{h}$ is a metric on $L$.

\end{theorem}

Clearly $h[x]=1$ iff $x$ covers $O$. Such elements are called \textit{atoms} of the lattice. The atoms of the lattice are analogues of the singletons in lattice of sets (one dimensional spaces in the lattice of subspaces). 

\begin{definition} A \textit{geometric modular lattice} is a modular
lattice of finite height in which every element is a join of atoms.
Further, if the geometric modular lattice is distributive, then it
will be called \textit{geometric distributive}. \end{definition}

The lattice of subsets is a an example of a geometric distributive lattice and the lattice of subspaces is an example of a geometric modular lattice.

\begin{figure}
 \centering
 \includegraphics[scale=0.5, trim=95 90 10 30, clip=true]{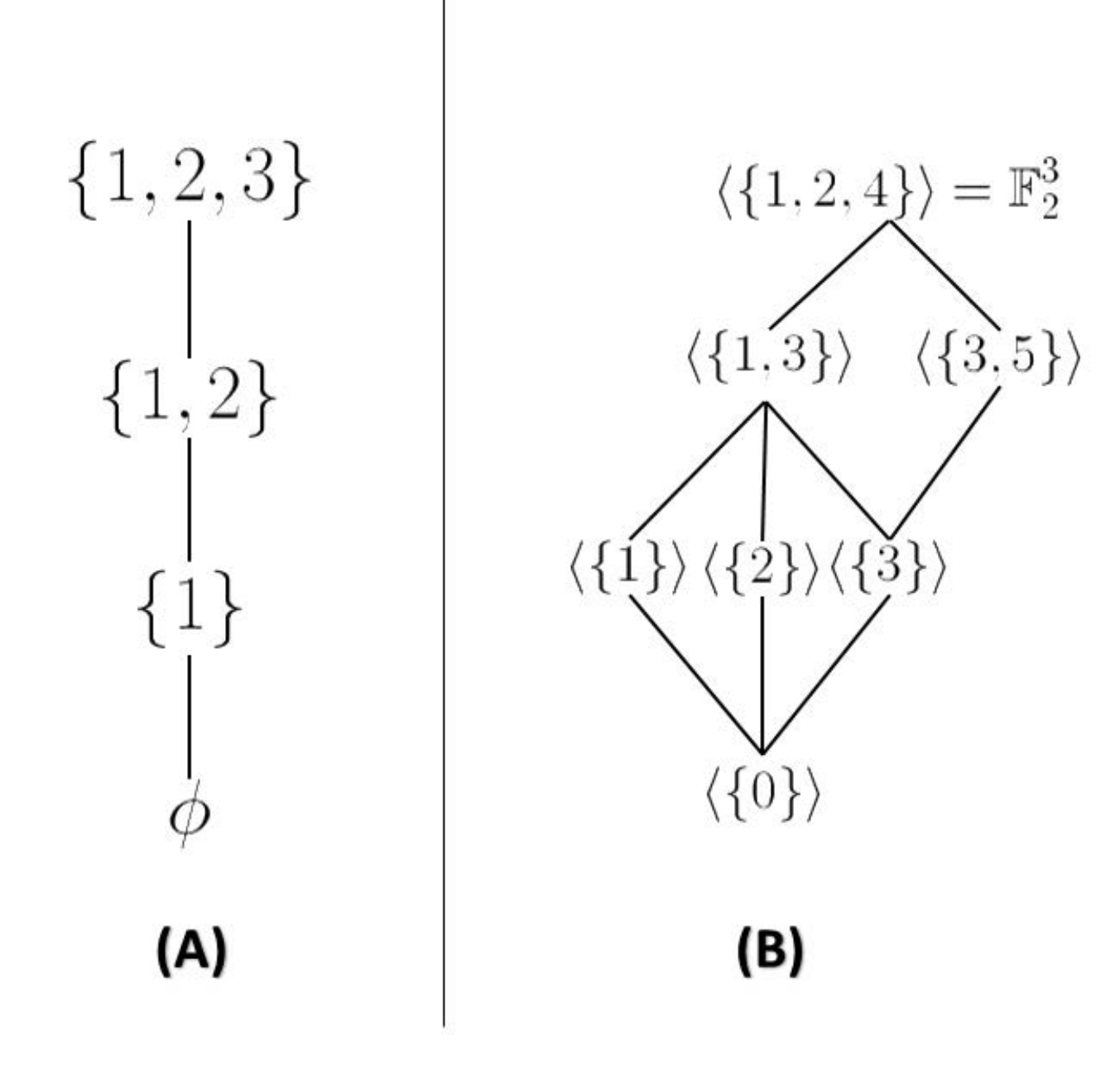}
 \caption{Two non-geometric lattices: In part (A), a non-geometric distributive lattice $L_1$ is shown. In part (B), a non-geometric, non-distributive modular lattice $L_2$ is depicted.}
 \label{fig_nongeometric}
\end{figure}

However, all modular (distributive) lattices are not geometric. For example, consider the sublattice (in fact, a chain) \[L_1 = \{\phi, \{1\}, \{1,2\}, \{1,2,3\}\}\] of $($\textit{Pow($\{1,2,3\}$)}$,\cup,\cap)$ (shown in part (A) of Fig. \ref{fig_nongeometric}). $L_1$ is not geometric because $\{1,2,3\}$ cannot be obtained as a join of atoms of $L_1$ (since there is only one atom in $L_1$). $L_1$ is distributive since it is a sublattice of distributive lattice viz. $($\textit{Pow($\{1,2,3\}$)}$,\cup,\cap)$. Therefore $L_1$ is an example of non-geometric distributive lattice.

In order to construct a non-geometric non-distributive modular lattice, we will consider a sublattice of $(\text{Sub}(\mathbb{F}_2^3),+,\cap)$. For the sake of presentation, we will represent a vector $(a,b,c)$ by the natural number $a + 2b + 4c$. For example, $(1,0,1)$ will be represented as $5$. Consider the sublattice 
\[L_2 = \{\langle\{0\}\rangle,\langle\{1\}\rangle,\langle\{2\}\rangle,\langle\{3\}\rangle, \langle\{1,3\}\rangle, \langle\{3,5\}\rangle,\mathbb{F}_3^2 \}\] of $(\text{Sub}(\mathbb{F}_2^3),+,\cap)$. $L_2$ is modular since it is a sublattice of a modular lattice viz. $(\text{Sub}(\mathbb{F}_2^3),+,\cap)$. $L_2$ is not distributive because it contains a copy of $M_3$ as a sublattice ($M_3$ is non-distributive). It can be verified that $\langle\{3,5\}\rangle$ cannot be obtained as a join of atoms of $L_2$ and thus $L_2$ is non-geometric. Therefore $L_2$ is an example of non-geometric non-distributive modular lattice.

\section{Lattice Schemes}

\label{sec_lat_sch}

In order to develop a lattice based framework for Singleton bounds, we need a definition of a code in this framework. In this section, we define `Lattice Schemes' which will serve as analogues of codes. This idea is motivated by the coding-like theory, introduced in \cite{Braun}. We will show that Hamming space, and the projective spaces are examples of lattice schemes. Henceforth, all the lattices are assumed to be geometric modular of finite height unless otherwise mentioned. 

A lattice scheme, which is analogous to a code, is defined as follows:
\begin{definition}
Let $L$ be a lattice and $d_h$ be the metric induced by the height function $h$ of the lattice $L$. A {\it lattice scheme} $C$ in $(L,d_h)$ is a subset of $L$ and the {\it minimum distance of $C$}, denoted by $d$ is defined as \[\displaystyle d:= \min_{a,b \in C, a \neq b} d_h(a,b).\] The dimension of a lattice scheme is defined as $n:=h(I)$.
\end{definition}

A coding space $(X,d_X)$ is a metric space where $X$ is a set and $d_X$ is a metric on $X$. A code $C$ in a coding space $(X,d_X)$ is a subset of $X$. The connection between lattice schemes and codes is made precise in the following definition.
\begin{definition}
\label{def_transform}
Let $C$ be a lattice scheme in $(L,d_h)$ and $\tilde{C}$ be a code in a coding space $(X,d_{X})$. We say that the {\it code} $\tilde{C}$ is {\it equivalent to a lattice scheme} $C$, if there exists a function $T:X \to L$, that satisfies the following conditions:

\begin{enumerate} \item $T(\tilde{C}) = C$ \item $d_h(T(a),T(b)) = d_X(a,b)$ for all $a, b \in \tilde{C}$ \end{enumerate} 

$T$ is called a transform for the code $\tilde{C}$.
\end{definition}

\begin{remark}
A transform $T$ of a code is one-one. To prove this fact, assume $T(a) = T(b)$, then $d_h(T(a),T(b)) = 0$. By 2) of Definition \ref{def_transform}, this would imply $d_X(a,b) = 0$. And since $d_X$ is a metric, we have $a=b$. 
\end{remark}

When a lattice scheme is equivalent to a code, we also say the code is equivalent to the scheme. A transform for a code preserves the distance between any pair of codewords. Therefore whenever a lattice scheme is equivalent to a code, the lattice scheme will have the same minimum distance as the code. The following proposition follows from Definition \ref{def_transform},

\begin{proposition} 
\label{prop_dist_preserve}

Let $C$ be a lattice scheme with minimum distance $d$ that is equivalent to code $\tilde{C}$ with minimum distance $\tilde{d}$, then

\begin{enumerate} \item $d=\tilde{d}.$ \item if $\tilde{A} \subseteq \tilde{C}$, then there exists $A \subseteq C$ such that the lattice scheme $A$ is equivalent to the code $\tilde{A}.$ \end{enumerate} \end{proposition} 
\begin{proof} Since `$C$ is equivalent to $\tilde{C}$', the first part follows from the definition. For the second part, since $C$ is equivalent to $\tilde{C}$, there exists a transform $T$. We define $A := T(\tilde{A})$. Clearly $A \subseteq C$ and the function $T$ still serves as a lattice transform for the code $\tilde{A}$. Thus $A$ is equivalent to $\tilde{A}$. \end{proof}

Due to the above proposition, given a lattice scheme equivalent to a code, we can talk of the minimum distance without specifying whether it is the minimum distance of the code or the minimum distance of the scheme. From the second part of the Proposition \ref{prop_dist_preserve}, we can infer that the subsets of codes are equivalent to certain subsets of schemes. Therefore, if we prove that a coding space itself is equivalent to some lattice, then any code in the coding space is equivalent to a scheme in the lattice. We use this observation to establish that every binary code is equivalent to a scheme in the power set lattice.

\begin{example} \label{exmpl_binary_codes}

Let $X = \{1,2,3,...,n\}$, $L = (${\it Pow($X$)}$,\cup,\cap)$ and $h(A) = |A|$. $L$ is a geometric distributive lattice and $d_h(A,B) = |A \cup B| - |A \cap B| = |A \triangle B|$ as seen in the previous section. Consider codes in the coding space $(\mathbb{F}_2^n,d_H)$ where $d_H$ is the Hamming distance between two vectors. 

We claim that the entire coding space $\mathbb{F}_2^n$ is equivalent to the power set lattice $L$. To see this let, \begin{align*} \phi: & \mathbb{F}_2 ^n \to \textit{Pow}(X) \\ & x \longmapsto \text{support}(x). \end{align*}
It can be verified that $\phi (x + y) = \phi(x) \triangle \phi(y)$ (where $\triangle$ represents the symmetric difference of sets) and that $\phi$ is onto. Further, $d_h(\phi(a),\phi(b))=| \phi(a) \triangle \phi(b)| = | \phi(a+b)|$. The number of elements in $\phi(a+b)$ will be the Hamming weight of $a+b$. Thus $| \phi(a+b)| = d_H(a+b,0) = d_H(a,b)$. Therefore, $\phi$ is the power set lattice transform for the binary code. 

Since the map is onto, $\phi$ is a bijective map. By application of the second part of Proposition \ref{prop_dist_preserve}, we see that every binary code is equivalent to a power set lattice scheme. 
 \end{example}

The lattice of subspaces, discussed in the previous section, also provide examples of lattice schemes. In a projective lattice, the metric induced by the height function is the subspace distance. Therefore, any subspace code is equivalent to a lattice scheme in the projective lattice.

\begin{example} Let $V$ be a vector space over $\mathbb{F}_q, L = ({\cal P}(V),+,\cap)$ is a projective lattice with height function $h(A) = \text{dim}(A)$ . The coding space is $(\text{Sub}(V),d_S)$ where $d_S$ represents the subspace distance defined in \cite{KoeKschi}. For any two subspaces $A$ and $B$, $d_S(A,B):= \dim(A+B) - \dim(A \cap B)$. Since $d_h(A,B) = h(A \vee B) - h(A \wedge B) = \text{dim}(A+ B) - \text{dim}(A \cap B)$, the metric induced by the height function is the subspace distance. The identity map can be a transform in this case. And thus, subspace schemes are equivalent to subspace codes. \end{example}

\section{Singleton Bound}
\label{sec_main_res}
We derive the Singleton bound for geometric modular lattices in this section. We use the notion of puncturing a scheme from \cite{KendSch} and investigate the effects of puncturing on the minimum distance of a lattice scheme. It will be proved that, after puncturing a scheme in a geometric modular lattice, the maximum drop in minimum distance will be two. However, if the lattice is known to be distributive, it is shown that the maximum drop in minimum distance is one. This observation will be applied repeatedly until the minimum distance drops to zero so that a bound can be derived on the cardinality of the scheme.

We will need the following definition of {\it Whitney number of the second kind}, from \cite{Stanley}, to state the lattice Singleton bound: \begin{definition} The {\it Whitney numbers} $W_k(L)$ of a lattice $L$ in a lattice with height $h$ is defined as \[ W_k(L) = \left|\{a \in L | h(a) = k\}\right|. \] \end{definition}

The Whitney numbers of a lattice count the total number of elements in the lattice of a given height.

\begin{definition} A scheme $C$ in $L$ is said to be {\it punctured to $C'$} if $C' = \{w \wedge a | a \in C \}$ for some $w \in L$. If $w$ has a height of $h(I)-1$, the scheme $C$ is said to be {\it punctured by a dimension}. \end{definition}


We need the following lemma (called the `distance drop lemma') to establish the proof of the main theorem later.

\begin{lemma}[Distance drop lemma] \label{lem_drop_dist}

Let $C$ be a scheme in $(L,d_h)$ with minimum distance $d$, and let $C$ be punctured by a dimension to $C'$, then

\begin{enumerate}

\item $L$ is distributive $\implies$ $d_{\min}(C') \geq d - 1$

\item In general, $ d_{\min}(C') \geq d - 2.$

\end{enumerate} \end{lemma} \begin{proof} 
Let $\tilde{a} := a \wedge w$, $\tilde{b} := b \wedge w$ and $\widetilde{a \vee b} := (a \vee b) \wedge w$.

Proof of 1): We assume $L$ is distributive. We have to show that \[d_h(\tilde{a} ,\tilde{b}) \geq d - 1 \] for all $a,b \in C$. By the definition of $d_h$, \begin{align*} d_h(\tilde{a} ,\tilde{b}) = h[\tilde{a} \vee \tilde{b}] - h[\tilde{a} \wedge \tilde{b}]. \end{align*} Since the lattice is distributive, we can write it as, \begin{align*} d_h(\tilde{a} ,\tilde{b}) = h[\widetilde{a \vee b}] - h[a \wedge b \wedge w]. \end{align*} The height function is a valuation and thus satisfies $h[x \vee y] + h[x \wedge y] = h[x] + h[y]$. We use this in the above equation to obtain, \begin{align*} d_h(\tilde{a} ,\tilde{b}) = h[a \vee b] + h[w] - h[a \vee b \vee w] - h[a \wedge b \wedge w]. \end{align*} Since $a \vee b \vee w \leq I$, it must be that $h[a \vee b \vee w] \leq h[I] = n$. Using this inequality and $h[w] = n-1$, we get \begin{align*} d_h(\tilde{a} ,\tilde{b}) \geq h[a \vee b] + (n-1) - n - h[a \wedge b \wedge w]. \end{align*} Clearly $a \wedge b \wedge w \leq a \wedge b$ and therefore $h[a \wedge b \wedge w] \leq h[a \wedge b]$. Using this, the definition of $d_h(a,b)$ and the fact that $d$ is the minimum distance of the scheme $C$, we finally get \begin{align*} d_h(\tilde{a} ,\tilde{b}) \geq d_h(a,b) - 1 \geq d - 1. \end{align*}

Proof of 2): We have to show that $d_h(a \wedge w ,b \wedge w) \geq d - 2$ for all $a,b \in C$. 
Again by the definition of $d_h$, 
\begin{align*} 
d_h(\tilde{a} ,\tilde{b}) = h[\tilde{a} \vee \tilde{b}] - h[\tilde{a} \wedge \tilde{b}].
\end{align*}
Repeatedly using the fact that the height function is a valuation and thus satisfies $h[x \vee y] + h[x \wedge y] = h[x] + h[y]$, we get, 
\begin{align*} d_h(\tilde{a} ,\tilde{b}) = h[a] + h[b] + 2h[w] - h[a \vee w] \\ - h[b \vee w] - 2h[a \wedge b \wedge w]. \end{align*} 
Clearly $a \wedge b \wedge w \leq a \wedge b$ and therefore $h[a \wedge b \wedge w] \leq h[a \wedge b]$. 
Additionally $\tilde{a}\leq I$ and $\tilde{b} \leq I$, which means $h[\tilde{a}], h[\tilde{b}] \leq h[I] = n$. Using this inequality and $h[w] = n-1$, we get the following:
\begin{align*} d_h(\tilde{a} ,\tilde{b}) \geq h[a] + h[b] + 2(n-1) - 2n - 2h[a \wedge b] \end{align*}
Using the definition of $d_h(a,b)$ and the fact that $d$ is the minimum distance of the scheme $C$, we finally get, 
\begin{align*} d_h(\tilde{a} ,\tilde{b}) \geq d_h(a,b) - 2 \geq d - 2. \end{align*} \end{proof}

The distance drop lemma states that in a non-distributive lattice, the drop in the minimum distance after puncturing a dimension, can be at most two units. So it would be interesting to know if it is possible that the drop of two units is exhibited by some scheme in a non-distributive lattice. The following example constructs such a scheme.

\begin{example} \label{example_proj_space} Let $V$ be a three dimensional space, over $\mathbb{F}_2$, spanned by $\{e_1,e_2,e_3\}$. Let $A_1 = $ $\langle\{e_1,e_2\}\rangle$, $A_2 = $ $\langle\{e_2,e_3 + e_1\}\rangle$ and $W = $ $\langle\{e_2,e_3\}\rangle$. We have $d_S(A_1,A_2) = 2$ but $d_S(W \cap A_1, W \cap A_2) = 0$. Note that $A_2 \cap (A_1 + W) \neq A_2 \cap A_1 + A_2 \cap W$ (that is, the sub-lattice generated by $A_1$, $A_2$ and $W$ is $M_3$) as expected. If our scheme contained $A_1, A_2$ and $W$, then puncturing the scheme with $W$ would have left $W$ as it is and punctured only the remaining two subspaces.
\end{example}

The above example clearly illustrates that the lack of distributivity in a lattice can drop the distance of a punctured scheme by two units. In fact, whenever a scheme has two elements, $a$ and $b$, which along with a third lattice element $c$ generates a sub-lattice isomorphic to $M_3$, we can puncture by a dimension so that the distance between $a$ and $b$ drop by two units. This is established in the Lemma \ref{lem_mod_drop_emthree}.

\begin{figure}
 \centering
 \includegraphics[scale=0.5, trim=90 100 50 10, clip=true]{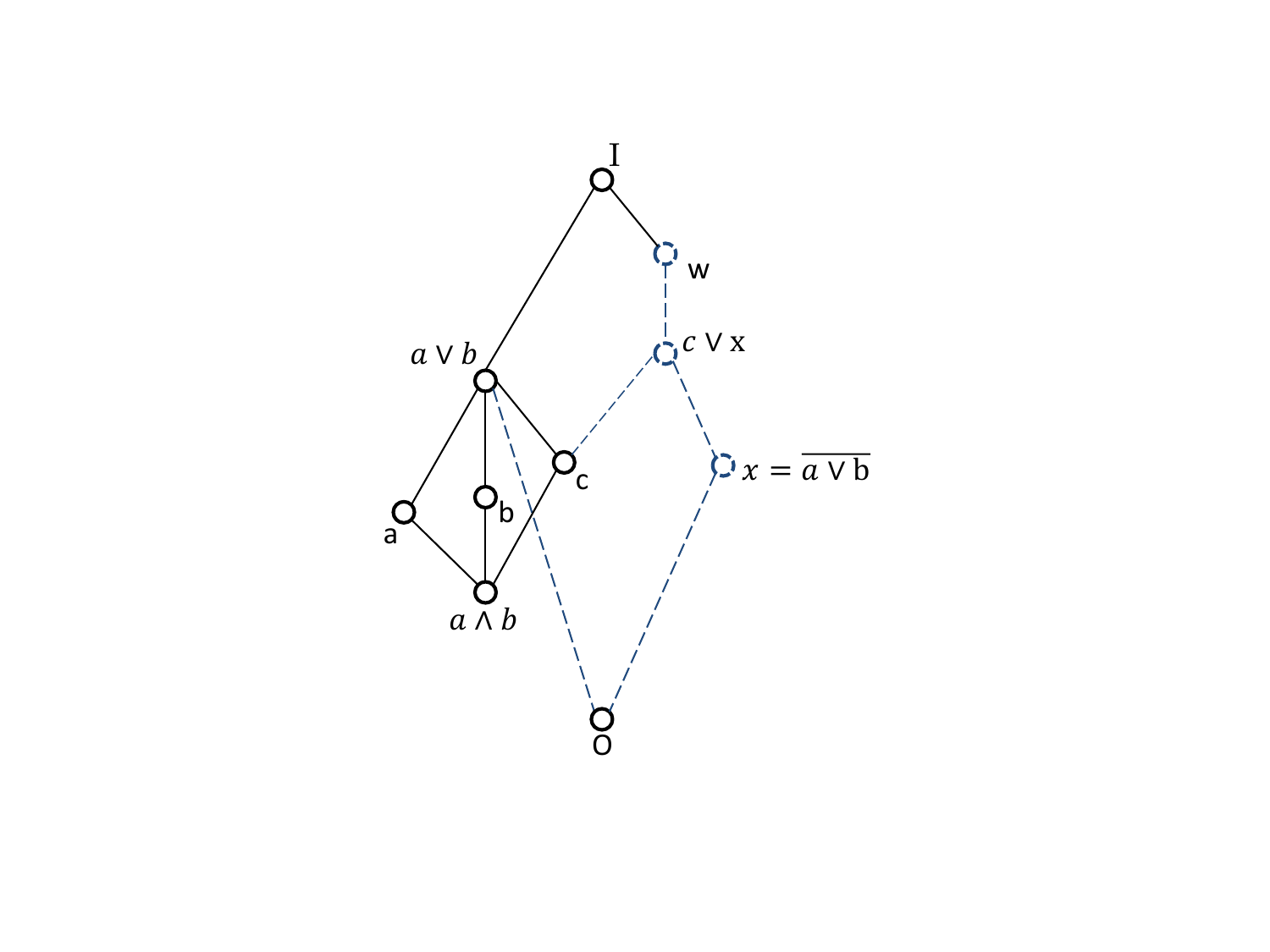}
 \caption{The black solid lines and bubbles denote the given data and the blue dashed lines and bubbles represent the elements constructed in the proof of Lemma \ref{lem_mod_drop_emthree}.}
 \label{fig_mthreestructproof}
\end{figure}

\begin{lemma}
\label{lem_mod_drop_emthree}
If the sub-lattice generated by $a,b,c \in L$ is isomorphic to $M_3$, then there exists a $w \in L$ with $h[w] = h[I]-1$ such that $d(a\wedge w,b \wedge w) = d(a,b) - 2$. 
\end{lemma}
\begin{proof}
We will refer to $|d(a,b) - d(a\wedge w,b \wedge w)|$ as `drop in distance'. From the proof of the second part of Theorem \ref{lem_drop_dist}, the drop in distance is two units if and only if all the inequalities in the that proof are equalities. That is, the drop in distance is two units if and only if \begin{align} \label{eqn_droptwo_1} a \vee w = b \vee w = I,\\ \label{eqn_droptwo_2} a \wedge b \leq w.\end{align} 

It is given that the sub-lattice generated by $\{a,b,c\}$ is isomorphic to $M_3$. In other words,
\begin{align}a \vee b = b \vee c = a \vee c,\\ a \wedge b = b \wedge c = a \wedge c.\end{align} 

By the `complementarity' property of geometric modular lattice $L$, there exists an $x \in L$ such that $(a \vee b) \wedge x = O$ and $(a \vee b) \vee x = I$. $x$ is called \emph{a complement} of $a \vee b$ (See \cite[pg.89]{Birkhoff}). Using the fact that $a \vee b = b \vee c$ and $c \leq b \vee c$, we get \begin{align} h[x] &= h[I] - h[a \vee b] \\ &= h[I] - h[b \vee c] \\ &< h[I] - h[c].\end{align}
But this means \begin{align}h[c \vee x] &\leq h[c] + h[x]\\ &< h[c] + (h[I] - h[c]) \\ &= h[I].\end{align} Therefore, since $L$ is a geometric lattice, we can find a $w$ such that $c \vee x \leq w$ and $h[w] = h[I] - 1$. If we pick this $w$ to puncture the lattice $L$ by a dimension, we will show that the drop in distance is two units by proving the relations (\ref{eqn_droptwo_1}) and (\ref{eqn_droptwo_2}). Fig.\ref{fig_mthreestructproof} shows a depiction of the complement $x$ of $a \vee b$ and the construction of $w$.

To establish \ref{eqn_droptwo_1}, we use a height calculation:
\begin{align*}
h[a \vee w] &= h[a \vee (c \vee x)] \\ &= h[(a \vee c) \vee x]\\ &= h[(a \vee b) \vee x]\\ &= h[I].
\end{align*}
Therefore we establish that $a \vee w = I$. In a similar manner, we can establish $b \vee w = I$. Since $a \wedge b \leq c$ and $c \leq c \vee x = w$, \ref{eqn_droptwo_2} is also proved. Hence the drop in distance is two units between $a$ and $b$.
\end{proof}

Next, we show that this happens in all non-distributive lattices since every non-distributive lattice has a $M_3$ sub-lattice.

\begin{theorem}
\label{thm_mod_drop}
There exists elements $a,b,w \in L$ with $h[w] = h[I]-1$ such that $d(a\wedge w,b \wedge w) = d(a,b) - 2$ if and only if $L$ is a non-distributive lattice. 
\end{theorem}
\begin{proof}
Again, we will refer to $|d(a,b) - d(a\wedge w,b \wedge w)|$ as `drop in distance'. From the first part of Theorem \ref{lem_drop_dist}, the drop in distance is at most one unit in a distributive lattice. This means that if the drop in distance is two units, the sub-lattice $M$ generated by $\{a,b,w\}$ is modular non-distributive and therefore $L$ is non-distributive. 

On the other hand, suppose $L$ is a non-distributive lattice. Then by Theorem \ref{thm_mthree} $L$ contains a sub-lattice $M$, isomorphic to $M_3$, generated by $\{a,b,c\}$ for some $a,b,c \in L$. Therefore, by Lemma \ref{lem_mod_drop_emthree}, there exists a $w \in L$ with $h[w] = h[I]-1$ such that $d(a\wedge w,b \wedge w) = d(a,b) - 2$. \end{proof}

We will now derive a Singleton bound for lattice schemes, that establishes an upper bound on the cardinality of the scheme, for a given minimum distance for geometric modular lattices.

\begin{theorem}[Lattice Singleton Bound(LSB)] \label{theorem_LSB} 
If $(L, d_h)$ is a geometric modular lattice with height $h$ and $h(I) = n$, $d_h$ is the metric induced by the height function, and $C$ is a scheme of $L$ with minimum distance $d$, then 
\begin{equation} |C| \leq |L'| \end{equation}
 where $L' = [0,w]$ for an element $w \in L$ and $h(w) = n - \alpha_L$. The constant $\alpha_L$ depends on the lattice as follows:
\begin{enumerate} 
\item $\alpha_L = d-1$, when $L$ is {\it distributive},
\item $\alpha_L = \lfloor {\frac{d-2}{2}} \rfloor$, when $L$ is {\it modular}.
\end{enumerate} 
\end{theorem} 
\begin{proof} 
Given the lattice $L$, pick an element $w'_1$ of height $n-1$ (which always exists in a geometric modular lattice). Now puncture the lattice by a dimension to get a new geometric modular lattice $L'$ with height $n-1$ and suppose that drop in minimum distance of the scheme, after puncturing a dimension, is at most $\beta$. If $d'$ is the minimum distance of the punctured scheme $C'$, then $d' \geq d - \beta$. If $d - \beta > 0$, then all elements in the punctured scheme $C'$ are still distinct. We repeat the puncturing operating on the new lattice. Suppose $w'_i$ is used to puncture the at the '$i$'th step, then the height of $w'_i$ in the lattice is $n-i$. The puncturing is repeated maximum number of times so that the minimum distance of the punctured scheme does not drop to zero. In other words, we keep puncturing the lattice, until the minimum distance of the punctured scheme is just short of zero. Let us say that the lattice was punctured $\alpha _L$ times. Since the minimum distance of the punctured scheme is still non zero, the punctured scheme contains exactly the same number of elements as $C$. At this stage, the minimum distance of the scheme is non-zero and all the elements of the punctured scheme are in an interval of $[0,w]$ where $w$ is an element of height $n - \alpha _L$. Clearly the number of elements in the scheme $C$ is upper bounded by the total number of elements in the punctured lattice. Thus 
\begin{align*} |C| \leq |L'|. \end{align*}
If the minimum distance of the scheme, after being punctured $\alpha _L$ times, is $D$, then $D \geq d - \beta \alpha _L$ and $D$ is the smallest number such that $D > 0$. This implies $\alpha_L$ is the largest number such that $\beta \alpha _L < d$. From the Lemma \ref{lem_drop_dist}, we know that $\beta = 1$ for a distributive lattice and $\beta = 2$ for a modular lattice. Thus $\alpha _L = d - 1$ for a distributive lattice and for a modular lattice, $\alpha_L = \lfloor {\frac{d-2}{2}} \rfloor$. This completes the proof. \end{proof}

According to Theorem \ref{thm_mod_drop}, when puncturing by a dimension, the distance between elements of non-distributive lattice scheme will decrease by at most two units. However, if the distance between $a$ and $b$ in a scheme is the minimum distance of the scheme and the sub-lattice generated by $\{a,b,w\}$ is distributive, then the drop in the distance between $a$ and $b$ is not more than one unit. In the proof of Theorem \ref{theorem_LSB}, when we repeatedly puncture by a dimension, we use a drop of two units to obtain the bound. Therefore if there is a non-distributive lattice scheme such that the elements that are at a minimum distance after each puncture form a non-distributive sub-lattice with the puncturing element, then the Singleton bound would be tighter for non-distributive lattices. However, we have not been able to construct such schemes.

We will now apply Theorem \ref{theorem_LSB} to two important lattice schemes (namely the power set lattice and the projective lattice) and derive the corresponding LSB. The LSB that we obtain coincides with the Singleton bound found in the literature. First, we derive the classical Singleton Bound in as a corollary to Theorem \ref{theorem_LSB}. 

\begin{corollary} Let $C$ be a code in $(\mathbb{F}_2 ^n, d_H)$, with minimum distance $d$, then $\displaystyle |C| \leq 2^{n-d+1}. $ 

\end{corollary} \begin{proof} In Example \ref{exmpl_binary_codes}, $c_L(n-\alpha _L, k)$ is the number of subsets of size $k$ in the scheme $C$, given that $v(I) = n - \alpha_L$. Using the fact that $L$ is geometric and distributive, $\alpha_L = d-1$. First we observe that, \[ c_L(n-d+1, k) \leq |\{A \in L | |A| = k\}| \] and \[ |\{A \in L | |A| = k\}| = \binom{n-d+1}{k}. \] Applying the LSB theorem, we get, 

\[ |C| \leq \sum_{k=0}^{n-d+1} \binom{n-d+1}{k} = 2^{n-d+1} \] which is the Singleton bound for $(\mathbb{F}_2 ^n, d_H)$. \end{proof}

A new bound for non-constant dimension subspace codes can be derived by applying Theorem \ref{theorem_LSB}. This Singleton bound in projective spaces will be specified using Gaussian numbers (they are the $q$-analogues of binomials \cite{LintWils}). 

\begin{corollary}[Singleton Bound for projective spaces] Let $C$ be a code in $(\text{Sub}(V), d_S)$ (where $V$ is vector space over $\mathbb{F}_q$), with minimum distance $d$, then 

\[\displaystyle |C| \leq \sum_{k=0}^{n-\lfloor \frac{d-2}{2} \rfloor}{n-\lfloor \frac{d-1}{2} \rfloor \brack k}_q \] where ${n \brack k}_q$ denotes the Gaussian number. \end{corollary} \begin{proof} The $L = (\text{Sub}(V),d_v)$ lattice with $v(A) = \dim(A)$ is a modular geometric lattice. We note that the Gaussian number ${N \brack K}_q$ is the total number of $K$-dimensional subspaces of an $N$-dimensional space over $\mathbb{F}_q$ and thus by Theorem \ref{theorem_LSB},

\[ |C| \leq \sum_{k=0}^{k=n-\alpha_L} {n- \alpha_L \brack k}_q \] where $\alpha_L = \lfloor \frac{d-2}{2} \rfloor$. \end{proof}  

\begin{figure*}[t]
 \centering
 \includegraphics[scale=0.5, trim=20 75 40 10, clip=true]{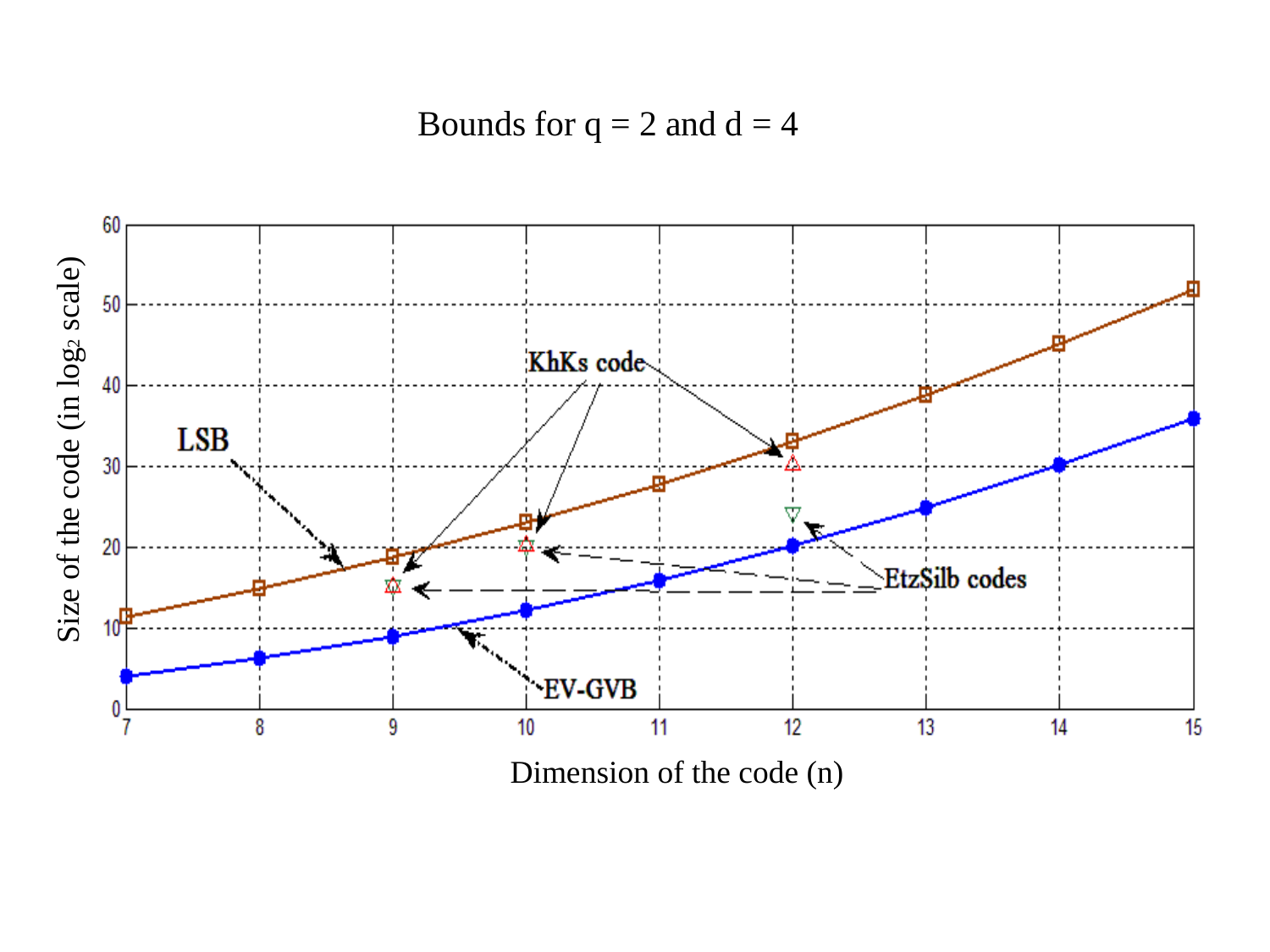}
 \caption{A plot that compares code sizes of different non-constant dimension codes and a lower bound on the code size. Note that the KhKs code sizes are within three bits from our upper bound. The subspaces are assumed to be over $\mathbb{F}_2$ and the minimum distance is fixed at $4$.}
 \label{fig_drawing}
\end{figure*}

Although this is a new upper bound, the tightness of the bound is not apparent. Therefore we investigate the behavior of our bounds, at such ranges, by plotting it with other bounds in the literature. We will compare our bound (LSB) to the Gilbert Varshamov bound (EV-GVB) proposed in \cite{EtzVar} for various values of the minimum distance. Further, we will plot points achieved by various codes in the literature.

The plot is shown in Fig. \ref{fig_drawing}. The minimum distance of the projective code and the finite field size has been fixed at $4$ and $2$ respectively, throughout this section. The plot clearly shows our upper bound above the lower bound EV-GVB. The points marked 'EtzSilb codes' and 'KhKs code' are the code parameters reported in \cite{EtzSilb2009} and \cite{KhalKschi} respectively, for $q=2$ and $d=4$. 

Fig. \ref{fig_drawing} shows that the EtzSilb codes and KhKs codes are close to optimal for $q=2$ and $d=4$. The KhKs code sizes are roughly $3$ bits away from the upper bound.

\section{Conclusion}
We introduced the notion of lattice schemes which serve as analogues of codes. We showed that binary codes and projective codes are special cases of lattice schemes. We derived a general notion of Singleton bound for lattices from which the classical Singleton bound for binary codes was derived as a corollary. We have proved that in any non-distributive modular lattices, a distance drop of two can be achieved by choosing an appropriate puncturing element. We proved that puncturing a dimension gives a tight bound for distributive lattices but not for projective lattices. The upper bound for non-constant dimension projective codes is also obtained. It is demonstrated that this bound is tighter when the minimum distance is much smaller than the dimension of the ambient space of the code.

It is not known whether there are lattice schemes that achieve the LSB for any lattice. It is not clear whether the Singleton bounds for rank metric codes, non-binary codes and quantum codes can be included in this framework. It would be useful to investigate interesting lattice schemes other than projective codes and binary codes. In \cite{EtzSilb2009}, Ferrers diagrams are used to construct projective codes and a bound similar to Singleton bound for rank-metric codes is also derived. Ferrers diagrams form a non-geometric distributive lattice. Therefore generalization of the LSB for non-geometric lattices is also an interesting direction for further research. It would be interesting if the bound presented in \cite{EtzSilb2009} can be presented from the point of view of lattices.

\end{document}